\documentclass[journal]{IEEEtran} 
\pdfoutput=1

\makeatletter                   
\let\proof\@undefined
\let\endproof\@undefined
\makeatother                    
\usepackage[USenglish,english]{babel}
\usepackage[cmex10]{amsmath} 
\usepackage{amssymb}  
\usepackage{amsthm}

\usepackage{graphicx} 
\usepackage[font=footnotesize]{subfig} 
\usepackage{mathptmx} 
\usepackage{times} 

\usepackage{cite} 

\newtheorem{defn}{Definition}
\newtheorem{ass}{Assumption}
\newtheorem{lem}{Lemma}
\newtheorem{thm}{Theorem}
\newtheorem{prop}{Proposition}
\newtheorem{cor}{Corollary}
\newtheorem{ex}{Example}
\newtheorem{rem}{Remark}

\newcommand{\R}{\mathbb{R}}
\newcommand{\N}{\mathbb{N}}
\newcommand{\PP}{\mathbb{P}}
\newcommand{\PE}{\mathbb{E}}

\newcommand{\X}{\mathbb{X}}
\newcommand{\Z}{\mathbb{Z}}

\newcommand{\V}{\mathbb{V}}
\newcommand{\W}{\mathbb{W}}
\newcommand{\dd}{\mathrm{d}}

\title{
 Constraint-Tightening and Stability in Stochastic Model Predictive Control
}
\author{Matthias Lorenzen$^{1}$,
Fabrizio Dabbene$^{2}$,
Roberto Tempo$^{2}$,
and 
Frank Allg\"ower$^{1}$
\thanks{$^{1}$Institute for Systems Theory and Automatic Control, University of Stuttgart, 
Germany
{\tt\small \{matthias.lorenzen,\allowbreak frank.allgower\}@ist.uni-stuttgart.de}. 
The authors would like to thank the German Research Foundation (DFG) for financial support of the project within the Cluster of Excellence in Simulation Technology (EXC 310/2) at the University of Stuttgart. \newline
$^{2}$CNR-IEIIT, Politecnico di Torino, 
Italy
{\tt\small \{roberto.tempo,\allowbreak fabrizio.dabbene\}@polito.it}. 
This research was partially supported by the joint CNR-JST international lab COOPS. 
\newline
\copyright 2015 IEEE. Personal use of this material is permitted. Permission from IEEE must be obtained for all other uses, in any current or future media, including reprinting/republishing this material for advertising or promotional purposes, creating new collective works, for resale or redistribution to servers or lists, or reuse of any copyrighted component of this work in other works.
}%
}

\begin{document}
\maketitle

\begin{abstract}
  Constraint tightening to non-conservatively guarantee recursive feasibility and stability in Stochastic Model Predictive Control is addressed.
  Stability and feasibility requirements are considered separately, highlighting the difference between existence of a solution and feasibility of a suitable, a priori known candidate solution. 
  Subsequently, a Stochastic Model Predictive Control algorithm which unifies previous results is derived, leaving the designer the option to balance an increased feasible region against guaranteed bounds on the asymptotic average performance and convergence time. 
  Besides typical performance bounds, under mild assumptions, we prove asymptotic stability in probability of the minimal robust positively invariant set obtained by the unconstrained LQ-optimal controller.
  A numerical example, demonstrating the efficacy of the proposed approach in comparison with classical, recursively feasible Stochastic MPC and Robust MPC, is provided.
\end{abstract}


\section{Introduction}
It is well-known that a moving horizon scheme like Model Predictive Control (MPC) might incur significant performance degradation in the presence of uncertainty and disturbances. This fact was already recognized in early publications on dynamic programming, see for instance~\cite[Chapter 9.5]{Dreyfus1977_ArtAndTheoryOfDynProg}.
To cope with this disadvantage, in recent years Robust MPC has received a great deal of attention for linear systems~\cite{Kothare1996_RobustMPCUsingLMI,Mayne2005_RobustMPCofConstrLinSysWithBoundedDist} as well as for nonlinear systems~\cite{Magni2003_RobustMPCForNonlinearDiscrTimeSys,Bayer2013_DiscreteTimeIncrementalISSforRobustMPC,Lazar2008_ISSsuboptimalNMPCwithApplToDCDCconv}.
In many cases, a stochastic model can be formulated to represent the uncertainty and disturbance, as for instance in the case of inflow material quality and purity in a chemical process or wind speed and turbulence in aircraft or wind turbine control.
This fact, and the inherent conservativeness of robust approaches, has led to an increasing interest in Stochastic Model Predictive Control (SMPC). 
A probabilistic description of the disturbance or uncertainty allows to optimize the average performance or appropriate risk measures. 
Furthermore, allowing a (small) probability of constraint violation, by introducing so-called \emph{chance constraints}, seems more appropriate in some applications, e.g. meeting the demand in a warehouse or bounds on the temperature or concentrations in a chemical reactor. 
Besides, chance constraints lead to an increased region of attraction without changing the prediction horizon.
Still, hard constraints, e.g. due to physical limitations, can be considered in the same setup. 

The first problem in Stochastic MPC is the derivation of computationally tractable methods to propagate the uncertainty for evaluating the cost function and the chance constraints. Both are multivariate integrals, whose evaluation requires the development of suitable techniques.
A second problem in SMPC is related to the difficulty of establishing recursive feasibility. 
In order to have a well-defined control law, it is necessary to guarantee that the optimal control program, which is solved online, remains feasible at future sampling times if it is initially feasible. 
Indeed, in classical MPC, recursive feasibility is usually guaranteed through showing that the planned input trajectory remains feasible in the next optimization step. This idea is extended in Robust MPC by requiring that the input trajectory remains feasible for all possible disturbances. 

In Stochastic MPC, a certain probability of future constraint violation is in general allowed, which leads to significantly less conservative constraint tightening for the predicted input and state, because worst-case scenarios become very unlikely. 
However, in this setup, the probability distribution of the state prediction at some future time depends on both the current state and the time to go. 
Hence, even under the same control law, the violation probability changes from time $k$ to time $k+1$, which can render the optimization problem infeasible.

  The first problem, uncertainty propagation and tractable reformulation of chance constraints, has gained significant attention and different methods to evaluate exactly, approximate or bound the desired quantities have been proposed in the Stochastic MPC literature.
  An exact evaluation is in general only possible in a linear setup with Gaussian noise or finitely supported uncertainties as in~\cite{Benandini2012_StabMPCofStochConstrLinSys}.
  Approximate solutions include particle approach~\cite{Blackmore2010_ProbabilisticParticleControlApproxOfChanceConstrSMPC} or polynomial chaos expansion~\cite{Mesbah2014_StochasticNonlinearMPC}.

  Bounding methods with guaranteed probabilistic confidence include~\cite{Kanev2006RobustMPC,Calafiore2013_RMPCviaScenario}, where the authors use the so-called scenario approach to cope with the chance constraint and determine at each iteration an optimal feedback gain (\hspace{1sp}\cite{Kanev2006RobustMPC}) or feed-forward input (\hspace{1sp}\cite{Calafiore2013_RMPCviaScenario}), respectively. While this approach allows for nearly arbitrary uncertainty in the system, the online optimization effort increases dramatically and recursive feasibility cannot be guaranteed. In~\cite{Zhang2014_OnSampleSizeOfRandMPCwAppl,Schildbach2014_ScenarioApproachForSMPC} the authors use an online sampling approach as well, but show how the number of samples can be significantly reduced.
  For linear systems with parametric uncertainty,~\cite{Cheng2014_SMPCforSysWithMultAndAddDist} proposes to decompose the uncertainty tube into a stochastic part computed offline and a robust part which is computed online. The paper~\cite{Fleming2014_StochasticTubeMPCforLPVSysWithSetInclCond} computes online a stochastic tube of fixed complexity using a sampling technique, but a mixed integer problem needs to be solved online. In~\cite{Cannon2009_ProbConstrMPCforMultAndAddStochUncert} layered sets for the predicted states are defined and a Markov Chain models the transition from one layer to another.

  For linear systems with additive stochastic disturbance, the system is usually decomposed into a deterministic, \emph{nominal} part and an autonomous system involving only the uncertain part.
  The approaches can then be divided into (i) computing a confidence region for the uncertain part and using this for constraint tightening, see~\cite{Cannon2011_StochasticTubesinMPC} for an ellipsoidal confidence region, and (ii) directly tightening the constraints, given the evolution of the uncertain part, e.g.~\cite{Kouvaritakis2010_ExplicitUseOfProbConstr} and~\cite{Korda2011_StronglyFeasibleSMPC}.
  A slightly different approach is taken in~\cite{Zhang2013_SochasticMPC}, where the authors first determine a confidence region for the disturbance sequence, as well, but then employ robust optimization techniques.
  Using the same setup, in~\cite{Chatterjee2011_StochasticMPCVectorSpaceApproach} the focus is to guarantee bounded variance of the state under hard input constraints.

  The second problem, recursive feasibility, has seemingly attracted far less attention. The issue has been highlighted in~\cite{Primbs2009_StochRecedingHorizonContrOfConstrLinSysWithMultNoise} and a rigorous solution has been provided in~\cite{Kouvaritakis2010_ExplicitUseOfProbConstr,Cannon2011_StochasticTubesinMPC}, where ``recursively feasible probabilistic tubes'' for constraint tightening are proposed. Instead of considering the probability distribution $\ell$ steps ahead given the current state, the probability distribution $\ell$ steps ahead given \emph{any} realization in the first $\ell-1$ steps is considered. This essentially leads to a constraint tightening with $\ell-1$ worst-case and one stochastic prediction for each prediction time $\ell$.
  In~\cite{Korda2011_StronglyFeasibleSMPC} the authors propose to compute a control invariant region and to restrict the successor state to be inside this region. This procedure leads to a feasible region which is less restrictive, 
  but stability issues are not discussed.

 The main contribution of this paper is to propose a nonconservative Stochastic MPC scheme that is computationally tractable and guarantees recursive feasibility. This is achieved by introducing a novel approach which unifies the previous results, combining the asymptotic performance bound of~\cite{Kouvaritakis2010_ExplicitUseOfProbConstr} with the advantages of the least restrictive approach in~\cite{Korda2011_StronglyFeasibleSMPC}. 
  Unlike previous works, we explicitly study the case when the optimized input sequence does not remain feasible at the next sampling time and present a constraint tightening to bound this to a desired probability $\epsilon_f$.
  Recursive feasibility is guaranteed through an additional constraint on the first step.
  With $\epsilon_f=1$ a scheme similar to~\cite{Korda2011_StronglyFeasibleSMPC} and with $\epsilon_f=0$, SMPC with recursively feasible probabilistic tubes is recovered. 
  We introduce a constraint tightening, which allows the parameter $\epsilon_f$ to be used as a tuning parameter to balance convergence speed and performance against the size of the feasible region.
  Under mild assumptions, we prove stability in probability of the minimal robust positively invariant region obtained by the unconstrained LQ-optimal controller. 
  As suggested in~\cite{Mayne2015_RobustStochMPCRightDirection} the online algorithm is kept simple and the main computational effort is offline. 
  The resulting offline chance constrained programs are briefly discussed and an efficient solution strategy using a sampling approach is provided.

  The remainder of this paper is organized as follows. Section~\ref{sec:ProbSetup} introduces the receding horizon problem to be solved. In Section~\ref{sec:ConstrTightAndAlgo} the proposed finite horizon optimal control problem is derived, starting with a suitable constraint reformulation, followed by recursive feasibility considerations of the optimization problem and a candidate solution. The section concludes with a summary of the algorithm.
  The theoretical properties are summarized in Section~\ref{sec:properties}, where a performance bound and a stability result are derived. A discussion on constraint tightening concludes the section and demonstrates the advantages of the approach.
The computation of the offline constraint tightening is discussed in Section~\ref{sec:NumExample}, followed by numerical examples that underline the advantages of the proposed scheme. Finally, Section~\ref{sec:Concl} provides some conclusions and directions for future work. 

  Preliminary results have been presented in~\cite{Lorenzen2015_improvedConstrTighteningForSMPC}. Building on these results, methods to bound the probability that a suitable candidate solution remains feasible are introduced and the implication on system theoretic properties stability and performance are analyzed thoroughly. A discussion on how to deal with joint chance constraints is presented and the numerical example has been updated to support the theory.
  Related results for systems with parametric uncertainty have been presented in~\cite{Lorenzen2015_ScenarioBasedStochasticMPC}, where constraint tightening via offline uncertainty sampling is addressed.

\paragraph*{Notation}
The notation employed is standard. Uppercase letters are used for matrices and lower case for vectors. $[A]_j$ and $[a]_j$ denote the $j$-th row and entry of the matrix $A$ and vector $a$, respectively.
Positive (semi)definite matrices $A$ are denoted $A \succ 0$ ($A \succeq 0$) and  $\|x\|_A^2 = x^\top A x$.
The set $\N_{>0}$ denotes the positive integers and ${\N_{\ge 0} = \{0\} \cup \N_{>0}}$, similarly $\R_{> 0}$, $\R_{\ge 0}$.
The notation $\PP_k\{\mathcal A\} = \PP\{ \mathcal A | x_k \} $ denotes the conditional probability of an event $\mathcal A$ given the realization of $x_k$, similarly $\PE_k\{\mathcal A\} = \PE\{ \mathcal A | x_k \}$.
We use $x_k$ for the (measured) state at time $k$ and $x_{l|k}$ for the state predicted $l$ steps ahead at time $k$.
The set of cardinality $T$ of vectors $v_{0|k}$,\dots $v_{T-1|k}$ will be denoted by $\mathbf{v}_{T|k}$.
%
$A\oplus B = \{ a+b | a\in A, b\in B\} $, $A\ominus B= \{ a \in A | a+b \in A ~\forall b\in B\}$ denotes the Minkowski sum and the Pontryagin set difference, respectively.
To simplify the notation, we use the convention $\sum_{k=a}^{b} c_k = 0$ for $a>b$.

\section{Problem Setup} \label{sec:ProbSetup}
In this section, we first describe the system to be controlled and introduce the basic Stochastic Model Predictive Control algorithm.

\subsection{System Dynamics, Constraints, and Objective}
Consider the following linear, time-invariant system with state $x_k \in \R^n$, control input $u_k \in \R^{m}$ and additive 
disturbance $w_k \in \R^{m_w}$
\begin{equation}
  x_{k+1} = A x_k + B u_k + B_w w_k.
  \label{eqn:xsystem}
\end{equation}

The disturbance sequence $(w_k)_{k\in\N_{\ge 0}}$ is assumed to be a realization of a stochastic process $(W_k)_{k\in\N_{\ge 0}}$ satisfying the following assumption.
\begin{ass}[Bounded Random Disturbance]\label{ass:disturbance}
$W_k$ for $k=0,1,2,\dots$ are independent and identically distributed, zero mean random variables with distribution $\PP$ and support $\W$. The set $\W$ is bounded and convex.
\end{ass}

The system is subject to probabilistic constraints on the state
and hard constraints on the input
\begin{subequations}
  \begin{align}
    \PP\{ [H]_j x_{k+l} &\le [h]_j ~|~ x_k\} \ge 1-[\varepsilon]_j ~ & j\in[1,p],~ l \in \N_{> 0}\label{eqn:probConstraints} \\
    G u_{k+l} &\le g  \quad &l \in \N_{\ge 0} \label{eqn:inputConstraints}
  \end{align}
  \label{eqn:origConstraints}%
\end{subequations}
with $H \in \R^{p\times n}$, $G \in \R^{q\times m}$, $h \in \R^{p}$, $g \in \R^{q}$, $\varepsilon \in [0,1]^p$ and the assumption that $u_{k+l}$ is a measurable function in $x_{k+l}$.
Equation~\eqref{eqn:probConstraints} restricts to $[\varepsilon]_j$ the probability of violating the linear state constraint $j$ at the future time $k+l$, given the realization of the current state $x_k$.
In the following, the notation $\PP_k\{\mathcal A\} = \PP\{ \mathcal A | x_k \} $ denoting the conditional probability of an event $\mathcal A$ given the realization of $x_k$ will be used.

The control objective is to (approximately) minimize $J_\infty$, the expected value of an infinite horizon quadratic cost
\begin{equation}
  J_\infty = \lim_{t \rightarrow \infty} \PE \left\{\frac{1}{t} \sum_{i=0}^t x_i^\top Qx_i + u_i^\top Ru_i \right\}
  \label{eqn:infHorizonCost}
\end{equation}
with $Q\in \R^{n\times n}$, $Q \succ 0$, $R\in \R^{m\times m}$, $R \succ 0$.

\subsection{Receding Horizon Optimization}
To solve the control problem, a Stochastic Model Predictive Control algorithm is considered.
The approach consists of repeatedly solving an optimal control problem with finite horizon $T$, but implementing only the first control action.

As it is common in linear Robust and Stochastic MPC, e.g. ~\cite{Kouvaritakis2010_ExplicitUseOfProbConstr}, the state of the system, predicted $l$ steps ahead from time~$k$
\begin{equation*}
  x_{l|k}=z_{l|k}+e_{l|k}
\end{equation*}
is split into a deterministic, nominal part $z_{l|k} = \PE_k\left\{ x_{l|k} \right\}$ and a zero mean stochastic error part $e_{l|k}$. 
Let $K\in\R^{n}$ be a stabilizing feedback gain such that $A_{cl} = A+BK$ is Schur. 
A prestabilizing error feedback $\tilde u_k = K e_k$ is employed, which leads to the predicted input
\begin{equation}
  u_{l|k} = Ke_{l|k} + v_{l|k}
  \label{eqn:inputParam}
\end{equation}
with $v_{l|k}$ being the free SMPC optimization variables. Hence, the dynamics of the \emph{nominal system} and error are given by
\begin{subequations}
  \begin{align}
    z_{l+1|k} &= A z_{l|k} + B v_{l|k}  &z_{0|k} = x_k\label{eqn:zSys}\\
    e_{l+1|k} &= A_{cl}e_{l|k} + B_w W_{l+k}  &e_{0|k} = 0 \label{eqn:eSys}%
  \end{align}%
  \label{eqn:zAndeSys}%
\end{subequations}
where $e_{l|k}$ are zero mean random variables and $z_{l|k}$ are deterministic.

The finite horizon cost $J_T(\mathbf x_{T+1|k},\mathbf{u}_{T|k})$ to be minimized at time $k$ is defined as 
\begin{equation}
  J_T(\mathbf x_{T+1|k},\mathbf{u}_{T|k}) 
    = \PE_k\left\{ \sum_{l=0}^{T-1} \left( x_{l|k}^\top Qx_{l|k} + u_{l|k}^\top Ru_{l|k} \right) + x_{T|k}^\top P x_{T|k}\right\}
  \label{eqn:finiteHorizonCostFnc}
\end{equation}
where $P$ is the solution to the discrete-time Lyapunov equation $A_{cl}^\top P A_{cl}+  Q + K^\top R K = P$. 
The expected value can be computed explicitly, which gives a quadratic, finite horizon cost function in the deterministic variables $z_{l|k}$ and $v_{l|k}$
\begin{equation}
  J_T(\mathbf z_{T+1|k},\mathbf{v}_{T|k}) =
  \sum_{l=0}^{T-1} \left( z_{l|k}^\top Qz_{l|k} + v_{l|k}^\top Rv_{l|k} \right) + z_{T|k}^\top P z_{T|k} + c
  \label{eqn:finiteHorizonCostFncDet}
\end{equation}
where $c = \PE_k \left\{ \sum_{i=0}^{T-1} e_{i|k}^\top (Q +K^\top RK)e_{i|k} + e_{n|k}^\top Pe_{n|k} \right\}$ is a constant term which can be neglected in the optimization.

The prototype finite horizon optimal control problem to be solved online is given in the following definition, where the constraint sets $\Z_{l}$ and $\V_{l}$ are derived from the chance constraints~\eqref{eqn:origConstraints} and some suitable terminal constraint as described in the next section.
\begin{defn}[Finite Horizon Optimal Control Problem]
  Given the system dynamics~\eqref{eqn:zAndeSys}, cost~\eqref{eqn:finiteHorizonCostFncDet} and nominal constraint sets $\Z_l$, $\V_l$ and $\Z_f$, the SMPC finite horizon optimization problem 
  is
  \vspace{-0.5\baselineskip}
  \begin{subequations}
    \begin{align}
      \min_{\mathbf{z}_{T+1|k}, \mathbf{v}_{T|k}} ~& J_T(\mathbf z_{T+1|k},\mathbf{v}_{T|k}) \label{seqn:optFnc}\\
      \text{s.t.} ~& z_{l+1|k} = A z_{l|k} + B v_{l|k},  \quad z_{0|k} = x_k  \nonumber\\
      ~& z_{l|k} \in \Z_{l}, ~ l\in[1,T] \label{seqn:protoConstr}\\
      ~& v_{l|k} \in \V_{l}, ~ l\in[0,T-1] \nonumber \\
      ~& z_{T|k} \in \Z_{f}. \nonumber
    \end{align}%
    \label{eqn:RecHorizonOptProg}%
  \end{subequations}%
  \vspace{-\baselineskip}
\end{defn}
The minimizer of~\eqref{eqn:RecHorizonOptProg}, which depends on the state $x_k$, is denoted $(z_{0|k}^*,\ldots,z_{T|k}^*,v_{0|k}^*,\ldots,v_{T|k}^*)$ and the SMPC control law is $u_k = v_{0|k}^*$.
The set of feasible decision variables for a given state $x_k$ is defined as
\begin{equation*}
  \mathbb D(x_k) = \left\{ \mathbf{z}_{T+1|k}, \mathbf{v}_{T|k} \in \R^{n(T+1)+mT} ~|~ \eqref{seqn:protoConstr}  \right\}.
\end{equation*}

In order to have a well-defined control law, it is necessary to ensure that, if initially feasible, the optimal control problem remains feasible at future sampling times, a property known as \emph{recursive feasibility}.
\begin{defn}[Recursive Feasibility]
  The finite horizon optimal control problem~\eqref{eqn:RecHorizonOptProg} is recursively feasible for system~\eqref{eqn:xsystem} under the SMPC control law $u_k = v_{0|k}^*$ if 
  \begin{equation*}
    \mathbb D(x_k) \neq \emptyset ~\Rightarrow~  \mathbb D(x_{k+1}) \neq \emptyset
  \end{equation*}
  for every realization $w_k \in \W$.
\end{defn}

The main goal 
is to suitably design the cost $J_T$ and constraint set $\Z_l$, $\V_l$ and $\Z_f$ of the finite horizon optimal control problem~\eqref{eqn:RecHorizonOptProg}
, such that in closed-loop the constraints~\eqref{eqn:origConstraints} are satisfied, recursive feasibility is ensured and the system is stabilized.

\section{Constraint Tightening and Stochastic MPC Algorithm} \label{sec:ConstrTightAndAlgo}
This section addresses the Stochastic MPC synthesis part. First, the deterministic, nonconservative constraint sets $\Z_l$ and $\V_l$ for the nominal system are derived, such that the constraints~\eqref{eqn:origConstraints} for system~\eqref{eqn:xsystem} hold in closed-loop under the SMPC control law.
These constraint sets are further modified to provide stochastic stability guarantees and recursive feasibility under all admissible disturbance sequences. We discuss the difference between existence of an a priori unknown feasible solution and feasibility of an a priori known candidate solution, which is unique to Stochastic MPC and plays a crucial role in proving stability.
A second constraint tightening is presented, where the probability of a given candidate solution being infeasible is a design parameter.
The section concludes with the resulting SMPC algorithm.

\subsection{Constraint Tightening} \label{ssec:ConstrTightening}
Given the evolution of the disturbance~\eqref{eqn:eSys}, similar to~\cite{Chisci2001_SystemsWithPersistentDisturbanceMPCwithrestrConstr, Kouvaritakis2010_ExplicitUseOfProbConstr}, we directly compute tightened constraints offline. However, we neither aim at the computation of recursively feasible probabilistic tubes nor at robust constraint tightening for the input.

\subsubsection*{State Constraints}
The probabilistic state constraints~\eqref{eqn:probConstraints} can non-conservatively be rewritten in terms of convex, linear constraint sets $\Z_l$ on the predicted nominal state $z_{l|k}$, as stated in the following proposition.
\begin{prop} \label{prop:constrSatisf}
  The system~\eqref{eqn:xsystem} satisfies the chance constraints~\eqref{eqn:probConstraints} for $k=1,\ldots,T$ and $j=1,\ldots,p$ if and only if
  the nominal system~\eqref{eqn:zSys} satisfies the constraints $z_{l|k} \in \Z_l$ with
  \begin{equation}
    \Z_l = \{z \in \R^n ~|~ H z \le \eta_{l} \} \quad l \in [1, T]
    \label{eqn:detStateConstraint}
  \end{equation}
  where $\eta_l$ is given by
  \begin{equation}
    \begin{aligned}
      {[\eta_{l}]_j} = \max_{\eta} &~ \eta\\
      \text{s.t.} &~ \PP_k\left\{ \eta \le [h]_j -  [H]_j  e_{l|k} \right\} \ge 1- [\varepsilon]_{j}, \quad j\in[1,p].
    \end{aligned}
    \label{eqn:offlineChanceConstrProgr}
  \end{equation}
\end{prop}
\begin{proof}
The constraint~\eqref{eqn:probConstraints} can be rewritten in terms of $z_{l|k}$ and $e_{l|k}$ as
\begin{equation}
  \PP_k\left\{  [H]_j z_{l|k} \le [h]_j -  [H]_j e_{l|k} \right\} \ge 1-[\varepsilon]_{j}
  \label{eqn:indProbConstraintsInZ}
\end{equation}
with $e_{l|k}$ being the solution to~\eqref{eqn:eSys}.
Equation~\eqref{eqn:indProbConstraintsInZ} is equal to $\exists \tilde \eta \in \R$ s.t. $[H]_j z_{l|k} \le \tilde \eta$ and $\PP_k\left\{  \tilde \eta \le [h]_j -  [H]_j e_{l|k} \right\} \ge 1-[\varepsilon]_{j}$. This is equal to $[H]_j z_{l|k} \le \eta$, with $\eta = \max_{\tilde \eta} \tilde \eta$ s.t. $\PP_k\left\{  \tilde \eta \le [h]_j -  [H]_j e_{l|k} \right\} \ge 1-[\varepsilon]_{j}$.
The maximum value exists as~\eqref{eqn:offlineChanceConstrProgr} can equivalently be written as
  \begin{equation*}
    \begin{aligned}
      {-[\eta_{l}]_j} = \min_{\eta} &~\eta \\
\text{s.t.} &~ \PP_k([H]_j e_{l|k} - [h]_j \le \eta  ) \ge 1 - [\varepsilon]_{j}.
    \end{aligned}
  \end{equation*}
  By Assumption~\ref{ass:disturbance} on the disturbance, the cumulative density function $F_{He-h}$ for the random variable $[H]_j e_{l|k} - [h]_j$ exists and is right-continuous. Using $F_{He-h}$, the constraint can be written as $ F_{He-h}(\eta) \ge 1 - [\varepsilon]_{j}$ which concludes the proof.
\end{proof}
Proposition~\ref{prop:constrSatisf} leads to $Tp$ independent, one dimensional, linear chance constrained optimization problems~\eqref{eqn:offlineChanceConstrProgr} that can be solved offline. Computational issues will be addressed in Section~\ref{ssec:SolvingSingleCC} and in the following the program~\eqref{eqn:offlineChanceConstrProgr} will be assumed to be solved. 
Note that the random variable $e_{l|k}$ does neither depend on the realization of the state $x_k$ at time $k$ nor at the optimization variables $\mathbf v_{T|k},~\mathbf z_{T+1|k}$. 

\subsubsection*{Input Constraints}
To decrease conservativeness, instead of a robust constraint tightening for the hard constraints on the input $u_k$, we propose a stochastic constraint tightening in the \emph{predictions}, which are restricted to optimal feed-forward instead of feedback control. In other words, we take advantage of the probabilistic nature of the disturbance and require that the (suboptimal) combination of SMPC feed-forward input sequence and static error feedback remains feasible for most, but not necessarily for all possible disturbance sequences. This is in line with the fact that at each sampling time the optimal input is recomputed and adapted to the actual disturbance realization, ensuring that the hard constraints~\eqref{eqn:inputConstraints} are satisfied.

Let $\epsilon_u \in [0,1)$ be a probabilistic level. Similarly to the state constraint tightening, we replace the original constraint~\eqref{eqn:inputConstraints} with $v_{l|k} \in \V_l$ where 
\begin{equation}
  \V_l = \left\{ v \in \R^m ~|~ G v \le \mu_l \right\} \quad l \in [0, T-1]
  \label{eqn:detInputConstraint}
\end{equation}
and $\mu_l$ is given by the solutions to $qT$ one dimensional, linear chance constrained optimization problems
\begin{equation}
  \begin{aligned}
    {[\mu_l]_j} = \max_{\mu} &~ \mu\\
    \text{s.t.} &~ \PP_{k}\left\{ \mu \le [g]_j -  [G]_j  K e_{l|k} \right\} \ge 1- \varepsilon_u, \quad j\in[1,q].
  \end{aligned}
  \label{eqn:offlineChanceConstrProgrU}
\end{equation}
We remark that in closed-loop, the hard input constraints~\eqref{eqn:inputConstraints} will be satisfied as $\mu_0 = g$.

\subsubsection*{Terminal Constraint}
We first construct a recursively feasible admissible set under a local control law and then employ a suitable tightening to determine the terminal constraint $\Z_f$ for the nominal system.
\begin{prop}[Terminal Constraint]
  For the system~\eqref{eqn:xsystem} with input $u_k=Kx_k$ let ${\X_f=\{x ~|~ H_f x \le h_f\}}$ be a (maximal) robust positively invariant polytope\footnote{
    For an in depth theoretical discussion, practical computation and polytopic approximations of $\X_f$ see~\cite{Blanchini1999_SetInvarianceInControl} for an overview or~\cite{Kolmanovsky1998_TheoryAndComputationOfDisturbanceInvariantSets} for details.}
    inside the set
  \begin{equation*}
    \tilde \X_f = \left\{ x ~|~ H A_{cl} x \le \eta_{1},~ G K x \le g \right\}
  \end{equation*}
  with $\eta_1$ according to~\eqref{eqn:offlineChanceConstrProgr}.
  For any initial condition in $\X_f$ the constraints~\eqref{eqn:origConstraints} are satisfied in closed-loop operation with the control law $u_k=Kx_k$ for all $k \ge 0$.
\end{prop}
\begin{proof}
  By definition, the set $\X_f$ is forward invariant for all disturbances and constraint~\eqref{eqn:inputConstraints} holds for all $x_k \in \X_f$.
  Furthermore
  \begin{equation*}
    \PP_k\{ [H]_j x_{k+1} \le [h]_j ~|~ x_{k} \} \ge 1-[\varepsilon]_j \quad \forall j\in[1,p]
  \end{equation*}
  is satisfied for all states $x_{k} \in \X_f$, which is sufficient for~\eqref{eqn:probConstraints}.
\end{proof}

To define the terminal constraint $\Z_f$ for the nominal system, a constraint tightening approach similar to~\eqref{eqn:offlineChanceConstrProgr} is needed. Let $\epsilon_f \in [0,1)$ be a probabilistic level, we define the terminal region
  \begin{equation}
    \Z_f = \{ z \in \R^{n} ~|~ H_f z \le \eta_f \}
    \label{eqn:termConstr}
  \end{equation}
  with
 \begin{equation}
  \begin{aligned}
    {[\eta_f]_j} = \max_{\eta} &~ \eta\\
    \text{s.t.} &~ \PP_{k}\left\{ \eta \le [h_f]_j -  [H_f]_j  e_{T|k} \right\} \ge 1- \varepsilon_f.
  \end{aligned}
  \label{eqn:offlineChanceConstrProgrTermConstr}
\end{equation} 

\subsection{Recursive Feasibility}\label{ssec:recFeas}
As it has been pointed out in previous works, e.g.~\cite{Primbs2009_StochRecedingHorizonContrOfConstrLinSysWithMultNoise, Kouvaritakis2010_ExplicitUseOfProbConstr}, the probability of constraint violation $\ell$ steps ahead at time $k$ is not the same as $\ell-1$ steps ahead at time $k+1$ given the realization of state $x_{k+1}$, in particular
\begin{equation*}
  Hz_{l|k}\le \eta_l \quad \nRightarrow \quad Hz_{l-1|k+1}\le \eta_{l-1}.
\end{equation*}
Hence, the tightened constraint sets~\eqref{eqn:detStateConstraint},~\eqref{eqn:detInputConstraint} and~\eqref{eqn:termConstr} do not guarantee recursive feasibility. 

A commonly used approach to recover recursive feasibility \emph{and} prove stability, is to use a mixed worst-case/stochastic prediction for constraint tightening. In~\cite{Kouvaritakis2010_ExplicitUseOfProbConstr, Korda2011_StronglyFeasibleSMPC} the constraint~\eqref{eqn:detStateConstraint} is replaced by
\begin{equation*}
  H z_{l|k} \le \eta_1 - \max_{w_i \in \W} \sum_{i=1}^{l-1}HA_{cl}^{i}w_i.
\end{equation*}

In~\cite{Korda2011_StronglyFeasibleSMPC} the authors point out that this approach is rather restrictive and leads to higher average costs if the optimal solution is ``near'' a chance constraint. 
Alternatively, if \emph{only} recursive feasibility is of interest, the authors propose to use a constraint only on the first input, to obtain a recursively feasible optimization program which 
is shown to be least restrictive.

In the following, we propose a hybrid strategy: We impose a first step constraint to guarantee recursive feasibility \emph{and} the previously introduced stochastic tube tightening with terminal constraint and cost to prove stability. At the cost of further offline reachability and controllability set computation, the proposed approach has the advantage of being less conservative compared to recursively feasible stochastic tubes, but yet guaranteed to stabilize the system at the minimal positively invariant region.

Let
\begin{equation*}
  C_T = \left\{ \begin{bmatrix} z_{0|k} \\ v_{0|k}  \end{bmatrix} \in \R^{n+m} ~|~ 
    \begin{array}[h]{l}
      \exists v_{1|k}, \ldots, v_{T-1|k} \in \R^m\\
      z_{l+1|k} = A z_{l|k} + B v_{l|k} \\
      H z_{l|k} \le \eta_{l}, ~ l\in[1,T]\\
      G v_{l|k} \le \mu_{l}, ~ l\in[0,T-1]\\
      H_f z_{T|k} \le \eta_f
    \end{array}
\right\}
\end{equation*}
be the $T$-step set and feasible first input for the nominal system~\eqref{eqn:zSys} under the tightened constraints $\Z_l$, $\V_l$ and $\Z_f$. 
The set can be computed via projection or backward recursion~\cite{Gutman1987_AlgorithmToFindMaximalStateConstrSet}. $C_T$ defines the feasible states and first inputs of the finite horizon optimal control problem.

Since the projection onto the first $n$ coordinates $C_{T,x} = \operatorname{Proj}_x(C_T)$ is not necessarily robust positively invariant with respect to the disturbance set $\W$, it is important to further compute a (maximal) robust control invariant polytope $C_{T,x}^\infty$ with the constraint $(x,u) \in C_T$.
Let $C_{T,x}^{0} = C_{T,x}$ and
\begin{equation*}
  C_{T,x}^{i+1} = \left\{ x \in C_{T,x}^i ~|~ 
    \begin{array}[h]{l}
      \exists u\in \R^{m} \text{ s.t. }  (x,u) \in C_{T}\\
      A_{cl}x + Bu \in C_{T,x}^{i} \ominus B_w \W \\
    \end{array}
  \right\},
\end{equation*}
the set $C_{T,x}^\infty$ is defined through $C_{T,x}^\infty = \cap_{i=0}^{\infty} C_{T,x}^{i}$.
The basis of a standard algorithm to compute $C_{T,x}^\infty$ is given by recursively computing $C_{T,x}^{i}$ until $C_{T,x}^i = C_{T,x}^{i+1}$ for some $i\in \N$ which implies $C_{T,x}^\infty = C_{T,x}^{i}$.
The basic idea and analysis of the sequence $C_{T,x}^i$ have been presented in~\cite{Bertsekas1972_InfTimeReachability},\cite[Section 5.3]{Blanchini2015_SetTheoreticMethodsInControl}. 


\begin{rem}
  The computation of the sets $C_T$ and $C_{T,x}^\infty$ is a long-standing problem in (linear) controller design, which has gained renewed attention in the context of Robust MPC. Efficient algorithms to exactly calculate or to approximate those sets exist, e.g.~\cite{Kolmanovsky1998_TheoryAndComputationOfDisturbanceInvariantSets,Blanchini2015_SetTheoreticMethodsInControl}. 
Matlab implementations of those algorithms as part of a toolbox can be found in, e.g.~\cite{MPT3_Toolbox,Kerrigan2000_Thesis_RobustConstraintSatisfaction-InvariantSetsAndPredictiveContr}.
\end{rem}

\subsection{Recursive Feasibility of the Candidate Solution}\label{ssec:feasCandidateSol}
Given a feasible input trajectory at time $k$, a candidate solution for time $k+1$ is given by a ``shifted solution'', as it is common in Robust and Stochastic MPC~\cite{Chisci2001_SystemsWithPersistentDisturbanceMPCwithrestrConstr,Kouvaritakis2010_ExplicitUseOfProbConstr}.
\begin{defn}[Candidate Solution]
  Given a solution $\mathbf v_{T|k}$,
  the candidate solution $\mathbf {\tilde v}_{T|k+1}$ to the SMPC optimization~\eqref{eqn:RecHorizonOptProg} at time $k+1$ for $k\ge0$ is defined by
  \begin{equation}
    \tilde{v}_{i|k+1}= \left\{ 
      \begin{aligned}
        &v_{i+1|k}+KA_{cl}^iBw_k && \quad i = 0,\ldots, T-2 \\
        &K (z_{T|k}+ A_{cl}^{i}Bw_k) && \quad i = T-1.
      \end{aligned}
      \right.
      \label{eqn:candSol}
  \end{equation}
\end{defn}

To prove asymptotic stability, not only existence of a feasible solution at each time $k$ is of interest, but also feasibility of an explicitly given \emph{candidate solution} at time $k+1$.
In this subsection, based on Section~\ref{ssec:ConstrTightening}, a refined constraint tightening is defined, which allows to explicitly bound the probability of the candidate solution being feasible in the next time step.

Let $\varepsilon_f\in [0,1)$ and $\W_f\subset \W$ be a convex $1-\varepsilon_f$ confidence region for $W_k$, i.e. ${\PP\left\{ W_k \in \W_f \right\} \ge 1-\varepsilon_f}$.
For $j=1,\dots, p$, $l = 1,\dots, T$ define 
  \begin{equation*}
    [\tilde \eta_l]_j = \min_{i=0,\dots,l-1} \{ [\hat \eta_{i,l}]_j + [\eta_{l-i}]_j \}
  \end{equation*}
with 
\begin{equation*}
  \left[ \hat \eta_{i,l} \right]_j = \min_{w_k \in \W_f} -[H]_j \sum_{\kappa=1}^{i}A_{cl}^{l-\kappa} B_w w_\kappa.
\end{equation*}
Similarly, define $\tilde \nu_l$ by replacing $H,h$ with $GK,g$, respectively.
For the terminal constraint 
let 
\begin{equation*}
  \begin{aligned}
    {[\hat \eta_{f,i}]_j} &= \min_{w_k \in \W_f} - [H_f]\sum_{\kappa = 1}^{i} A_{cl}^{T-\kappa} B_w w_\kappa,\\
    [\eta_{f,i}]_j &= \max_{\eta} \eta, ~\text{s.t. } \PP_k\{ \eta \le [h_f]_j - [H_f]_j e_{i|k}  \}
  \end{aligned}
\end{equation*}
to define 
$[\tilde \eta_f]_j = \min_{i=0,\dots,T} \{ [\hat \eta_{f,i}]_j + [\eta_{f,T-i}]_j \}$.

Tightened constraints, where the probability of infeasibility of the candidate solution can be specified a priori, are obtained by replacing $\eta_l$, $\nu_l$ and $\eta_f$ with $\tilde \eta_l$, $\tilde \nu_l$ and $\tilde \eta_f$.
\begin{prop}[Recursive Feasibility of the Candidate Solution]
  Let the state, input and terminal constraints in~\eqref{eqn:RecHorizonOptProg} be given by
  \begin{equation}
    \begin{aligned}
      \Z_l &= \{z \in \R^n ~|~ H z \le \tilde\eta_{l} \}, \quad l\in[1,T]\\
      \V_l &= \{v \in \R^m ~|~ G v \le \tilde\mu_{l} \}, \quad l\in[0,T-1]\\
      \Z_f &= \{z \in \R^n ~|~ H_f z \le \tilde\eta_f \}
    \end{aligned}%
    \label{eqn:constrEpsf}
  \end{equation}
  with $\Z_f \subseteq \Z_T$.

  If it exists $\mathbf z_{T+1|k}$ such that $(\mathbf v_{T|k},\mathbf z_{T+1|k})\in \mathbb D(x_k)$ then, with probability no smaller than $1-\epsilon_f$, $(\mathbf{\tilde v}_{T|k+1},\mathbf{\tilde z}_{T+1|k+1})\in \mathbb D(x_{k+1})$ with $\tilde z_{i+1|k+1}= A \tilde z_{i|k+1} + B \tilde v_{i|k+1}$ and $\tilde z_{0|k+1} = x_{k+1}$.
\end{prop}

\begin{proof}
  With probability  $1-\epsilon_f$ it holds $w_k \in \mathbb W_f$, hence it suffices to prove the claim for $w_k \in \mathbb W_f$.

  Assume $w_k \in \mathbb \W_f$, recursive feasibility of the terminal constraint follows from robust recursive feasibility of the terminal region and robust reachability for all $w_k \in W_f$. Furthermore $\tilde z_{T|k+1} \in \Z_T$ is implied by the assumption $\Z_f \subseteq \Z_T$.

  Constraint satisfaction for the state constraints for $l < T$ follows inductively
  \begin{equation*}
    \begin{aligned}
      &{[H \tilde z_{l|k+1}]_j} = [ H z_{l+1|k} + HA_{cl}^{l}B_w w_{k}]_j \le [\tilde \eta_{l+1} + HA_{cl}^{l}B_w w_{k}]_j\\
      &= \min_{i=0,\dots,l}\left\{ [\eta_{l+1-i}]_j - \max_{w_{\kappa} \in \W_f}\left[ [H]_j \sum_{\kappa=1}^{i} A_{cl}^{l+1-\kappa}B_w w_\kappa \right] \right. \\
      &\left. \hspace{2cm} + {[H]_j}A_{cl}^{l}B_w w_{k} \right\} \\
      &\le \min_{i=1,\dots,l}\left\{ [\eta_{l+1-i}]_j - \max_{w_{\kappa} \in \W_f}\left[ [H]_j \sum_{\kappa=2}^{i} A_{cl}^{l+1-\kappa}B_w w_\kappa \right] \right.\\
      &\left. \hspace{2cm} - \max_{w \in \W_f}\left[ {[H]_j} A_{cl}^{l}B_w w \right] + [H]_j A_{cl}^{l}B_w w_{k} \right\} \\
      &\le \min_{i=0,\dots,l-1}\left\{ [\eta_{l-i}]_j - \max_{w_{\kappa} \in \W_f}\left[ [H]_j \sum_{\kappa=1}^{i} A_{cl}^{l-\kappa}B_w w_\kappa \right]  \right\} = \tilde \eta_l
    \end{aligned}
  \end{equation*}
  for all $k$, $l$ and $j$. Similarly for the input, replacing $H$ and $\eta_l$ by $GK$ and $\mu_l$.
\end{proof}

While the constraints introduced in Section~\ref{ssec:ConstrTightening} only allow for an analysis of the probability of infeasibility of the candidate solution, 
the maximal probability is a design parameter when the constraints~\eqref{eqn:constrEpsf} are employed.
This alternative constraint tightening essentially closes the gap between ``recursively feasible probabilistic tubes''~\cite{Kouvaritakis2010_ExplicitUseOfProbConstr} which are recovered with $\varepsilon_f = 0$ and the ``least restrictive'' scheme presented in~\cite{Korda2011_StronglyFeasibleSMPC} where only existence of a solution is considered. 
The impact of $\varepsilon_f$ on the convergence and provable average closed-loop cost will be highlighted in the next section. The influence on the size of the feasible region is demonstrated in the example in Section~\ref{sec:NumExample}.

\subsection{Resulting Stochastic MPC Algorithm}
The final Stochastic MPC algorithm can be divided into two parts: (i) an offline computation of the involved sets and (ii) the repeated online optimization. 

\emph{Offline:} 
Determine the tightened constraint sets $\Z_l$, $\V_l$ and $\Z_f$ according to either \eqref{eqn:offlineChanceConstrProgr},~\eqref{eqn:offlineChanceConstrProgrU}, and~\eqref{eqn:offlineChanceConstrProgrTermConstr}, or~\eqref{eqn:constrEpsf}.
Determine the first step constraint $ C_{T,x}^\infty$ according to the section~\ref{ssec:recFeas}.

\emph{Online:} For each time step $k=0,1,2,\ldots$
\begin{enumerate}
  \item Measure the current state $x_k$,
  \item Solve the linearly constrained quadratic program~\eqref{eqn:RecHorizonOptProg} 
    with additional first step constraint $ C_T^\infty$
    , i.e.
      \begin{subequations}
      \begin{align}
        (\mathbf{z}_{T+1|k}^*,\mathbf{v}_{T|k}^*) ~&= \arg \min_{\mathbf{z}_{T+1|k},\mathbf{v}_{T|k}} J_T(\mathbf z_{T+1|k},\mathbf{v}_{T|k}) \label{seqn:MPCCost}\\
        \text{s.t.} ~& z_{l+1|k} = A z_{l|k} + B v_{l|k}  \quad z_{0|k} = x_k \nonumber\\
        ~& z_{l|k} \in \Z_l, ~ l\in[1,T] \nonumber\\
        ~& v_{l|k} \in \V_l, ~ l\in[0,T-1]   \label{seqn:MPCConstr}\\
        ~& z_{T|k} \in \Z_f \nonumber\\
        ~& z_{1|k} \in  C_{T,x}^\infty \ominus B_w\W, \nonumber
      \end{align}
      \label{eqn:MPCOptProg}
      \end{subequations}
  \item Apply $u_k = v^*_{0|k}$.
\end{enumerate}

\section{Properties of the proposed SMPC Scheme} \label{sec:properties}
In this section, we formally derive the control theoretic properties of the proposed SMPC scheme, 
in particular the influence of $\epsilon_f$, the bound on the probability of the candidate solution being infeasible. 
We first derive a bound on the asymptotic average state cost, which highlights the connection to~\cite{Kouvaritakis2010_ExplicitUseOfProbConstr} and proves bounded variance of the state. 
This is followed by a proof of asymptotic stability in probability of a robust invariant set, which is novel in Stochastic MPC and shows the connection to tube based Robust MPC approaches like~\cite{Chisci2001_SystemsWithPersistentDisturbanceMPCwithrestrConstr,Mayne2005_RobustMPCofConstrLinSysWithBoundedDist}. This asymptotic behavior has previously been claimed but only shown in simulations in~\cite{Deori2014_CompApproachesToRMPC}. 
The section concludes with a discussion on offline relaxation of chance constraints in Stochastic MPC.

\subsection{Asymptotic Average Performance}
Prior to a stability analysis, we prove recursive feasibility of the SMPC algorithm, which is provided by the following proposition.
\begin{prop}[Recursive Feasibility] \label{prop:recFeas}
  Let 
  \begin{equation*}
    \tilde{\mathbb D}(x_k) = \left\{ (\mathbf z_{T+1|k}, \mathbf v_{T|k}) \in \R^{n(T+1)+mT} ~|~ ~\eqref{seqn:MPCConstr} \right\}.
  \end{equation*}
  If $(\mathbf z_{T+1|k},\mathbf v_{T|k}) \in \tilde{\mathbb D}(x_k)$, then ${\tilde{\mathbb D}(x_{k+1}) \neq \emptyset}$ for every realization $w_k \in \W$ and ${x_{k+1} = A x_k + Bv_{0|k} + B_w w_k}$ .
\end{prop}
\begin{proof}
  From $z_{1|k} \in  C_{T,x}^\infty \ominus B_w\W$ it follows $x_{k+1} \in C_{T,x}^{\infty}$ and by construction $C_{T,x}^{\infty} \subseteq \{ x ~|~ \tilde{\mathbb D}(x) \neq \emptyset \}$.
\end{proof}

Due to the persistent excitation through the additive disturbance, it is clear that the system does not converge asymptotically to the origin, but ``oscillates'' with bounded variance around it.
The following theorem 
summarizes the constraint satisfaction and 
provides a bound on the asymptotic average stage cost. 
\begin{thm}[Main Properties]\label{thm:constrSatisAndAvgcost}
  If $x_0 \in C_{T,x}^\infty$, then the closed-loop system under the proposed SMPC control law satisfies the hard and probabilistic constraints~\eqref{eqn:origConstraints} for all future times $k$ and 
  \begin{equation*}
    \lim_{t\rightarrow \infty} \frac{1}{t} \sum_{k=0}^t \PE\left\{ \|x_k\|_Q^2 \right\} \le (1-\epsilon_f) \PE\left\{ \|B_w w\|_P^2 \right\} + \epsilon_f C
  \end{equation*}
  with $\epsilon_f$ the maximum probability that the previously planned trajectory is not feasible, $C = L~max_{w\in\W}\|B_w w\|$ and $L$ the Lipschitz constant of the optimal value function ${V(x_k) = J_T(\mathbf z_{T+1|k}^*,\mathbf{v}_{T|k}^*)}$ of~\eqref{eqn:MPCOptProg}.
\end{thm}
\begin{proof}
  Since, by Proposition~\ref{prop:recFeas}, the SMPC algorithm is recursively feasible, 
  chance constraint satisfaction follows from Proposition~\ref{prop:constrSatisf} and hard input constraint satisfaction from ${e_{0|k} = 0}$ and hence $\mu_0 = g$.

  To prove the second part, we use the optimal value of~\eqref{eqn:MPCOptProg} as a stochastic Lyapunov function.
  Let $V(x_k) = J_T(\mathbf z_{T+1|k}^*,\mathbf{v}_{T|k}^*)$ be the 
  optimal value function of~\eqref{eqn:MPCOptProg}, which is known to be continuous, convex and piecewise quadratic in $x_k$~\cite{Bemporad2002_ExplicitLQRforConstrSys}. Hence, a Lipschitz constant $L$ on $C_{T,x}^\infty$ exists.
  The old input trajectory does not remain feasible with at most probability $\epsilon_f$, but we can bound the cost increase of $V(x_{k+1})-V(z_{1|k})$ by $L~\max_{w\in\W} \|B_w w\|$.
  
  Let $\PE\left\{ V(x_{k+1}) | x_k, \mathbf{\tilde v}_{T|k+1}~\text{feasible}\right\}$ be the expected optimal value at time $k+1$, conditioning on the state at time $k$ and feasibility of the candidate solution $\tilde v_{l|k+1} = v_{l+1|k}^*+KA_{cl}^l B_w w_k$.
  \begin{equation*}
    \begin{aligned}
      & \PE\left\{ V(x_{k+1}) | x_k, \mathbf{\tilde v}_{T|k+1}~\text{feasible} \right\} - V(x_k) \\
      \le& \sum_{l=1}^{T-1} \left(  \|z_{l|k}^*\|_Q^2 + \|v_{l|k}^*\|_R^2  \right) + \|z_{T|k}^*\|_{(Q+K^\top R K)}^2 + \|z_{T+1|k}^*\|_P^2 ~ \\
      & + \PE\left\{ \sum_{l=1}^{T} \| A_{cl}^{l-1}B_w w_k\|_{(Q +K^\top RK)}^2 + \|A_{cl}^{T}B_w w_k \|_P^2 \right\} \\
      & -\left( \sum_{l=0}^{T-1} \left(  \|z_{l|k}^*\|_Q^2 + \|v_{l|k}^*\|_R^2  \right) + \|z_{T|k}^*\|_P^2 \right) \\
      =& \|z_{T|k}^*\|_{(Q+K^\top R K)}^2 + \|z_{T+1|k}^*\|_P^2 - \|z_{0|k}^*\|_Q^2 - \|v_{0|k}^*\|_R^2  - \|z_{T|k}^*\|_P^2  \\
      & + \PE\left\{ \|B_w w_k\|_P^2 \right\}  \\
      \le& - \|z_{0|k}\|_Q^2 + \PE\left\{ \|B_w w_k\|_P^2 \right\} = - \|x_{k}\|_Q^2 + \PE\left\{ \|B_w w_k\|_P^2 \right\}
    \end{aligned}
  \end{equation*}
  where $v_{l|k}^*$, $z_{l+1|k}^*$, $l=0,\ldots, T-1$ denote the optimal solution of~\eqref{eqn:MPCOptProg}, respectively predicted state at time $k$, and $z_{T+1|k}^* = (A+BK)z_{T|k}^*$. Note that the expected value of all $w$-$z$ cross-terms equals zero because of the zero-mean and independence assumption. Furthermore, since we defined the terminal cost as the solution to the discrete-time Lyapunov equation it holds that $A_{cl}^\top P A_{cl} + Q + K^{\top}RK = P$.
 
  Combining both cases we obtain by the law of total expectation
\begin{equation*}
  \begin{aligned}
    & \PE\left\{ V(x_{k+1}) | x_k \right\} - V(x_k) \\
    \le& (1-\epsilon_f)\left(\PE\left\{ V(x_{k+1}) | x_k,\mathbf{\tilde v}_{T|k+1}~\text{feasible} \right\} - V(x_k)\right) + \\
    & \epsilon_f \left(-\|x_{k}\|_Q^2 + ~L~\max_{w\in\W} \|B_w w\| \right) \\
    \le & -\|x_{k}\|_Q^2 + (1-\epsilon_f)\PE\left\{ \|B_w w_k\|_P^2 \right\} +  \epsilon_f C.
  \end{aligned}
\end{equation*}
The final statement follows by taking iterated expectations.
\end{proof}

\begin{rem}
  A terminal region, which is forward invariant with probability $\epsilon_f$, can be used instead of a robust forward invariant terminal region. In this case, Theorem~\ref{thm:constrSatisAndAvgcost} still holds. 
\end{rem}
\subsection{Asymptotic Stability}
In this section, we prove, under mild assumptions, the existence of a set $\X_\infty$ which is asymptotically stable in probability for the closed-loop system under the proposed Stochastic MPC algorithm. 
In particular, by the proposed SMPC control law, the same set is stabilized as with the Robust MPC proposed in~\cite{Chisci2001_SystemsWithPersistentDisturbanceMPCwithrestrConstr} or with the Stochastic MPC proposed in~\cite{Kouvaritakis2010_ExplicitUseOfProbConstr}. The different constraint tightening leads to a possibly different transient phase. 
The price to obtain a larger feasible region can be a longer convergence time before the terminal set is reached.

\begin{defn}[Asymptotic Stability in Probability]
A compact set $\mathbb{S}$ is said to be asymptotically stable in probability for system~\eqref{eqn:xsystem} with a control law $u_k=\kappa(x_k)$, if for each $\varepsilon \in \R_{>0}$ and $\rho \in [0,1)$ $\exists \delta \in \R_{>0}$ such that
\begin{equation*}
  \|x_0\|_{\mathbb{S}} \le \delta \Rightarrow  \PP\{ \sup_{k \ge 0} \|x_k\|_{\mathbb{S}} \ge \varepsilon \} \le 1- \rho
\end{equation*}
and for a neighborhood $\mathcal N_{\mathbb{S}}$ of $\mathbb{S}$, for all $\varepsilon_2 \in \R_{>0}$
\begin{equation*}
  x_0 \in \mathcal N_{\mathbb{S}} \Rightarrow \lim_{k'\rightarrow \infty} \PP\{ \sup_{k>k'} \|x_k\|_{\mathbb{S}} < \varepsilon_2 \} = 1
\end{equation*}
where $\mathcal N_{\mathbb{S}}$ is called region of attraction.
\end{defn}

To streamline the presentation, we make the following assumption on the control gain $K$, as well as two non-restrictive technical assumptions.
\begin{ass}\label{ass:basicStabAssumptions}
  ~
  \begin{itemize}
    \item The feedback gain $K$ for the prestabilizing and terminal controller is chosen to be the unconstrained LQ-optimal solution.
    \item Let $\X_\infty$ be the minimal robust positively invariant set for the system~\eqref{eqn:xsystem} with input $u_k = Kx_k$ and let $\mathbb B$ be an open unit ball in $\R^n$. 
      It exists $\lambda \in \R_{>0}$ such that $\X_\infty \oplus \lambda \mathbb B \subseteq \X_f$.
    \item The set $C_{T,x}^{\infty}$ is compact.
  \end{itemize}
\end{ass}

Under this assumption, the main result of this section, asymptotic stability of $\X_\infty$, can be formally stated.
\begin{thm}[Asymptotic Stability]\label{thm:asymptStability}
  Under Assumption~\ref{ass:basicStabAssumptions}, the 
  set $\X_\infty$ is asymptotically stable in probability with region of attraction $C_{T,x}^\infty$ for the system~\eqref{eqn:xsystem} with the proposed SMPC controller.
\end{thm}

We prove the theorem by first proving it under the assumption that the candidate solution remains feasible at each time step. Then, we prove that it exists a set $\mathbb S$ where this feasibility assumption is verified and that for every probability $\rho \in (0,1]$ and state $x_0$ in $C_{T,x}^\infty$ it exists a time $N\in \N_{\ge 0}$ such that $\PP\{x_N \in \mathbb S\} \ge 1-\rho$.

The proof differs from standard proofs using a stochastic Lyapunov function because of the nonzero probability that the candidate solution does not remain feasible during a transient phase.

The following lemma is inspired by Theorem 8 in~\cite{Chisci2001_SystemsWithPersistentDisturbanceMPCwithrestrConstr}, where Robust MPC is considered.
\begin{lem}\label{lem:asymptStabRobust}
  Given the system~\eqref{eqn:xsystem} with $x_0 \in C_{T,x}^\infty$ and the proposed SMPC controller.
  If Assumption~\ref{ass:basicStabAssumptions} holds and the candidate solution $\mathbf{\tilde v}_{T|k}$ remains feasible for all $k > 0$, then the state $x_k = \zeta_k + \xi_k$ can be separated into a part $\zeta_k$ and a part $\xi_k$, such that the origin is asymptotically stable for $\zeta_k$ with region of attraction $C_{T,x}^{\infty}$ and $\xi_k \in \X_\infty ~\forall k\ge0$.
\end{lem}
\begin{proof}
  Let 
  \begin{subequations}
    \begin{align}
      \zeta_{k+1} &=A_{cl}\zeta_k + B(u_k - K(\zeta_k + \xi_k)) \quad &&\zeta_0 = x_0 \label{eqn:zetaSys}\\
      \xi_{k+1} &= A_{cl} \xi_k + B_w w_k \quad &&\xi_0 = 0. \label{eqn:xiSys}
    \end{align}
  \end{subequations}
  Given recursive feasibility, in~\cite{Chisci2001_SystemsWithPersistentDisturbanceMPCwithrestrConstr} it has been shown that $c_k = u_k - K(\zeta_k + \xi_k)$ is bounded and $c_k \rightarrow 0$ for $k\rightarrow \infty$. Since $A_{cl}$ is Schur stable, the system~\eqref{eqn:zetaSys} is input to state stable (ISS) with respect to the input $c_k$ and hence $\zeta_k$ converges to the origin. Furthermore, for $x_k \in \X_f$ it holds that $c_k=0$, which together with Assumption~\ref{ass:basicStabAssumptions} implies asymptotic stability of the origin for system~\eqref{eqn:zetaSys}.
\end{proof}

\begin{cor} \label{cor:finiteTimeConvergence}
  If Assumption~\ref{ass:basicStabAssumptions} holds and the candidate solution $\mathbf{\tilde v}_{T|k}$ remains feasible for all $k > 0$, then there exists $N_\varepsilon$ such that $\|\zeta_k\|<\varepsilon$ for all $k \ge N_\varepsilon$. In particular, 
  it 
  exists $N_f \in \N_{\ge 0}$ such that $x_k \in \X_f$ for all $k\ge N_f$ and $x_0 \in C_{T,x}^\infty$.
\end{cor}
\begin{proof}
  From asymptotic stability, it follows that the origin is a uniform attractor for~\eqref{eqn:zetaSys} and hence
  for each $x_0 \in C_{T,x}^{\infty}$ exists a neighborhood $\mathcal N_{x_0}$ of $x_0$ and a $N_{x_0} \in \N_{>0}$ such that $\forall \zeta_{0} \in \mathcal N_{x_0}$ it holds $\zeta_k \in \lambda \mathbb B ~ \forall k>N_{x_0}$~\cite{Bhatia1967_DynamicalSystemsStabTheorieAndApplications}.
  The collection $\{\mathcal N_{x_0}\}_{x_0 \in C_{T,x}^{\infty}}$ is an open covering of the set $C_{T,x}^{\infty}$, hence by compactness of $C_{T,x}^{\infty}$ we can choose a finite subcollection $\{\mathcal N_{x_0}\}_{x_0 \in J}$ of $\{\mathcal N_{x_0}\}_{x_0 \in C_{T,x}^{\infty}}$ that also covers $C_{T,x}^{\infty}$.
  Letting $N_f = \max_{x_0 \in J} N_{x_0}$ it follows $\zeta_k \in \lambda \mathbb B ~ \forall k>N_f$ and all $x_0 \in C_{T,x}^{\infty}$. Using $\xi_k \in \X_\infty$, by Assumption~\ref{ass:basicStabAssumptions} this implies $x_k \in \X_f ~\forall k>N_f$ and $x_0 \in C_{T,x}^{\infty}$.
\end{proof}

It can be shown that the candidate solution remains feasible for all $k>k'$ if $x_{k'}$ is inside the terminal region. 
In this case, 
Lemma~\ref{lem:asymptStabRobust} holds and we only need to consider $x_k \notin \X_f$.
\begin{lem} \label{lem:invarianceTerminalRegion}
  The terminal region $\X_f$ is robust forward invariant for the closed-loop system under the proposed SMPC algorithm and 
  $( \mathbf{\tilde z}_{T+1|k+1}, \mathbf{\tilde v}_{T|k+1}) \in \tilde{\mathbb D}(x_{k+1})$ for all $k\ge k'$ if $x_{k'} \in \X_f$.
\end{lem}
\begin{proof}
  The unconstrained optimal solution to~\eqref{eqn:MPCOptProg} equals the control inputs generated by the LQR, $v_{l|k}^{*} = Kz_{l|k}^{*}$.
  For $z_{0|k} \in \X_f$ robust forward invariance of the terminal region implies constraint satisfaction of the unconstrained optimal solution
  \begin{equation*}
    \begin{aligned}
      &~ HA_{cl} (z_{l|k} + e_{l|k}) \le \eta_1 \quad \forall e_{l|k} \\
      \Leftrightarrow &~ H z_{l+1|k} \le \eta_1 - H A_{cl} e_{l|k}  \quad \forall e_{l|k} \\
      \Rightarrow &~ H z_{l+1|k} \le \eta_{l+1}
    \end{aligned}
  \end{equation*}
  and similarly for the input and terminal constraints.
  Hence, in the terminal region, the proposed SMPC controller equals the unconstrained LQR. Since $\X_f$ is robust forward invariant under the unconstrained LQ optimal controller the statement follows.
\end{proof}

Under Assumption~\ref{ass:basicStabAssumptions}, Lemma~\ref{lem:asymptStabRobust} and~\ref{lem:invarianceTerminalRegion} suffice for stability of the proposed algorithm. Before proving attractivity, we need another lemma.

\begin{lem} \label{lem:feasOfConsecutiveTimes}
  Let $I_{k'} = [k',k'+N_f-1]$ denote some interval of length $N_f$ and $\mathcal A_{k'}$ the event that the candidate solution $\mathbf{\tilde v}_{T|k}$ is feasible $\forall k\in I_{k'}$.
For each $\rho \in (0,1]$ there exists $N_{\rho} \in \N_{\ge 0}$
    \begin{equation}
      \PP\left\{ \cup_{k'=0}^{N_\rho} \mathcal A_{k'} \right\} \ge 1-\rho.
      \label{eqn:probFeasOfConsecutiveTimes}
    \end{equation}
\end{lem}
Lemma~\ref{lem:feasOfConsecutiveTimes} states, that for fixed probability $1-\rho$ we can find a sufficiently long horizon such that, at some point within this horizon, the candidate solution stays feasible $N_f$ consecutive times and hence, by Corollary~\ref{cor:finiteTimeConvergence}, enters the terminal region.
\begin{proof}
  Let $1-\varepsilon_f$ denote the probability that the candidate solution stays feasible in the next sampling instant.
  The left hand side of~\eqref{eqn:probFeasOfConsecutiveTimes} can be crudely over-approximated by the probability of staying feasible during one of the time periods $I_i =  [iN_f,(i+1)N_f-1]$ for $i\in[0,\lfloor \frac{N_\rho}{N_f} \rfloor]$.
For each $I_i$ we have $\PP\{\mathcal A_{iN_f}\}\ge (1-\varepsilon_f)^{N_f} =: 1- \beta_f$. Hence $\PP\left\{ \cup_{k'=0}^{N_\rho} \mathcal A_{k'} \right\} \ge \PP\left\{ \cup_{i=0}^{\lfloor \frac{N_\rho}{N_f} \rfloor} \mathcal A_{iN_f} \right\} \ge 1-(\beta_f)^{\lfloor \frac{N_\rho}{N_f} \rfloor + 1}$.
Since $\beta_f \in [0,1)$, the right hand side of the inequality is increasing with $N_\rho$ and converges to $1$.
\end{proof}

\begin{lem}[Attractivity]\label{lem:attractivity}
  Under Assumption~\ref{ass:basicStabAssumptions}, 
  for all $\varepsilon_2 \in \R_{>0}$
  \begin{equation*}
    x_0 \in C_{T,x}^{\infty} \Rightarrow \lim_{k'\rightarrow \infty} \PP\{ \sup_{k>k'} \|x_k\|_{\X_\infty} < \varepsilon_2 \} = 1.
  \end{equation*}
\end{lem}
By Corollary~\ref{cor:finiteTimeConvergence} and Lemma~\ref{lem:invarianceTerminalRegion} the closed-loop system converges if the candidate solution remains feasible for $N_f$ consecutive time-steps. By Lemma~\ref{lem:feasOfConsecutiveTimes}, for any given probability $\rho$, there exists a sufficiently long horizon $N_\rho$ such that this holds with probability $\rho$. We use the Borel-Cantelli Lemma and Fatou's Lemma~\cite{klenke2014_book_ProbabilityTheory} to show that in fact the probability grows \emph{fast enough}.\\
\begin{proof}
  Let $N=\max\{N_{\varepsilon_2}, N_f\}$ with $N_{\varepsilon_2},~N_f$ according to Corollary~\ref{cor:finiteTimeConvergence} and define the event ${\mathcal B_{k'} = \left\{ \sup_{k>k'} \|x_{k+N}\|_{\X_\infty} \ge \varepsilon_2 \right\}}$.
  By Corollary~\ref{cor:finiteTimeConvergence}, Lemma~\ref{lem:invarianceTerminalRegion} 
  it holds $\PP\{\mathcal B_{k'}\} \le 1 - \PP\left\{ \cup_{k=0}^{k'} \mathcal A_{k} \right\} $.
  Inserting the explicit bound derived in the proof of Lemma~\ref{lem:feasOfConsecutiveTimes} leads to
  \begin{equation*}
    \begin{aligned}
      \sum_{k=0}^{\infty} \PP\{\mathcal B_k\} 
      &\le \sum_{k=0}^{\infty} \beta_f^{ \frac{k}{N_f} } 
      < \infty.
    \end{aligned}
  \end{equation*}
  By the Borel-Cantelli Lemma we have that $\PP\left\{ \lim_{k\rightarrow \infty}\sup \mathcal B_k \right\} = 0$ and hence by Fatou's Lemma $\lim_{k\rightarrow \infty}\sup \PP\left\{ \mathcal B_k \right\} = 0$ which concludes the proof.
\end{proof}

\begin{proof}[Proof (Theorem~\ref{thm:asymptStability})]
  Stability follows from the robust case together with robust recursive feasibility in the terminal region (Lemma~\ref{lem:invarianceTerminalRegion}). Attractivity follows from Lemma~\ref{lem:attractivity}.
\end{proof}

A direct corollary of Theorem~\ref{thm:asymptStability} is a tighter bound on the asymptotic average performance.
\begin{cor}
  Under Assumption~\ref{ass:basicStabAssumptions} it holds
  \begin{equation*}
    \lim_{t\rightarrow \infty} \frac{1}{t} \sum_{k=0}^t \PE\left\{ \| x_k \|_Q^2 \right\} \le \PE\left\{ \| B_w w\|_P^2 \right\}.
  \end{equation*}
\end{cor}
\begin{proof}
  Let $\rho_f(k)$ be the probability that the candidate solution $\mathbf{\tilde v}_{T|k}$ is infeasible
  and $\Delta_V(k') = \PE\left\{ V(x_{k'}) \right\} - V(x_0)$.
  Analogous to the proof of Theorem~\ref{thm:constrSatisAndAvgcost} it holds
  \begin{equation*}
    \begin{aligned}
      &\Delta_V(k') \le \sum_{k=0}^{k'} \PE\{-\|x_k\|_Q^2 + \|B_w w_k\|_P^2\} + \rho_f(k)C \\
      &\frac{1}{k'}\sum_{k=0}^{k'} \PE\{\|x_k\|_Q^2 \} + \frac{\Delta_V(k') }{k'} \le \PE\{\|B_w w\|_P^2\} + \frac{C}{k'}\sum_{k=0}^{k'} \rho_f(k).
    \end{aligned}
  \end{equation*}
  Since $\sum_{k=0}^{k'} \rho_f(k) \le \sum_{k=0}^{k'} \varepsilon_f \PP\{B_{k-N-1}\} < \infty$ the result follows by letting $k'\rightarrow \infty$.
\end{proof}

\subsection{Discussion: Offline Relaxation of Chance Constraints}
In this section, we briefly discuss joint vs single chance constraints of the form
\begin{equation*}
  \begin{aligned}
    \PP\{ H x &\le h \} \ge 1-\varepsilon, \\
  \PP\{ [H]_j x &\le [h]_j \} \ge 1-[\varepsilon]_j,
  \end{aligned}
\end{equation*}
respectively, and contrast the approaches taken in~\cite{Zhang2013_SochasticMPC} and~\cite{Cannon2011_StochasticTubesinMPC} with our approach.

It is a known result in Robust Tube MPC, that the tightened constraint $\Z$ involves at most as many linear inequalities as the original constraint $\X$~\cite{Blanchini2015_SetTheoreticMethodsInControl}. For a joint chance constraint, an analogue result does not hold. A tightened constraint on $z$ can in general not be expressed by a finite number of linear constraints and in the cases in which this can be done (e.g., polytopic $\W$ with uniform distribution), the tightened constraint involves in general more linear inequalities than the original one, see Fig.~\ref{fig:jointCC} for an example.
\begin{figure}[htpb]
  \begin{center}
    \includegraphics{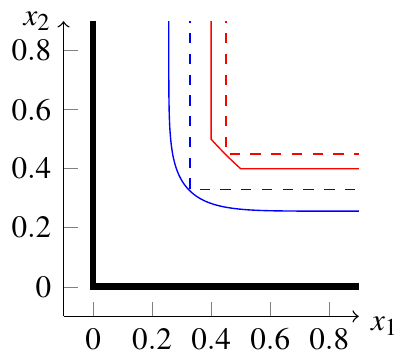}
  \end{center}
  \caption{\small Constraint $x_{1,2}\ge 0$ and tightened constraints for $z$ with $x=z+e$ and $e_{1,2}$ normal $e_i \sim \mathcal{N}(0,0.2^2)$ or uniform $e_i \sim \mathcal{U}(-0.5,0.5)$ distributed and violation probability $\varepsilon=0.1$. Dashed lines show tightened constraints for single chance constraints with $\varepsilon_j=\varepsilon/2$.}
  \label{fig:jointCC}
\end{figure}

Using single chance constraints to approximate joint chance constraints does not lead to an increased number of tightened constraints, but is in general either more conservative or increases the probability of constraint violation at certain points in state space.

We remark that the often proposed procedure of determining a confidence region for the uncertainty $w_k$ or the uncertain state $e_l$ and then requiring the constraints to hold for all realizations within these sets, leads to the same result. This does not approximate the true joint chance constraint but tighter versions of the original constraints, as well. Furthermore, choosing a parametrization for these sets increases the conservatism unless it fits the underlying distribution perfectly, e.g., spheroids for normal distributed uncertainties. In contrast, the approach taken here is tight for arbitrary distributions and constraints given by~\eqref{eqn:origConstraints}.

With the following illustrative example, we show the advantage of tightening each constraint according to a predefined probability of violating it, instead of jointly optimizing the tightening of all constraints with a given overall violation probability or similarly jointly determining the parameters for the confidence region of the uncertainty.
In particular, we show that the latter approach might lead to undesired, conservative results.
\begin{ex}\label{ex:jointCCParamDet}
  Consider the one dimensional case $x=z+e$ with $e$ having a non symmetric pdf, e.g., $e+0.2\sim\text{Gamma}(2,0.1)$ and constraint $\PP\left\{ |x| \le 1 \right\} \ge 1-\varepsilon$. 
Optimizing the constraint tightening jointly, i.e.
\begin{equation*}
  \begin{aligned}
    \max_{\eta}&~ \eta_1 + \eta_2 \\
    \text{s.t.} &~ \PP\left\{ \eta_1 \le 1-e,~ \eta_2 \le 1+e \right\} \ge 1-\varepsilon
  \end{aligned}
\end{equation*}
to derive the nominal constraints $-\eta_2 \le z \le \eta_1$
or optimizing the bounds of a confidence interval $[\gamma_1, \gamma_2]$ for $e$
\begin{equation*}
  \begin{aligned}
    \max_{\gamma}&~ \gamma_2 - \gamma_1 \\
    \text{s.t.} &~ \PP\left\{ \gamma_1 \le e \le \gamma_2 \right\} \ge 1-\varepsilon
  \end{aligned}
\end{equation*}
to derive the constraint $-1-\gamma_1 \le z \le 1-\gamma_2$
leads to a biased outcome.
As illustrated in Fig.~\ref{fig:constrExample}, 
with $\varepsilon=0.05$, the result is
$\eta_1 \approx 0.73$, $\eta_2 = 0.8$, respectively
$\gamma_1=-0.2$, $\gamma_2\approx 0.27$.
The constraint $x\ge0$ holds with zero probability of violation and $x\le0$ with $\varepsilon$ probability of violation.
If we maximize $x$, the result of the deterministic problem with tightened constraints will be equal to the solution of the original chance constrained problem. If we minimize $x$, the result will coincide with the solution of the robust problem.
\end{ex}
\begin{figure}[htpb]
  \centering
    \includegraphics{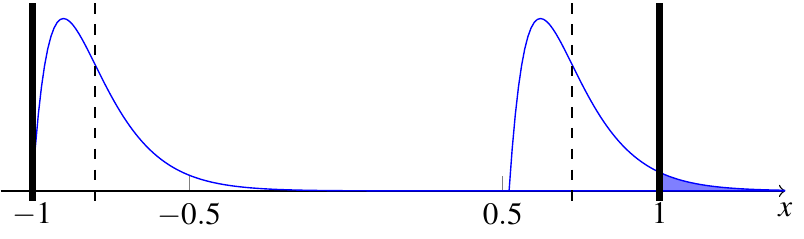}
  \caption{\small Constraints $|x| \le 1$ (solid line), resulting tightened constraints for the nominal state $z$ (dashed line) and probability density functions for $x$ with $z=-0.8$ and $z=0.73$, respectively (blue). A joint chance constraint evaluation as described in Example~\ref{ex:jointCCParamDet} leads to a biased outcome: The lower bound is satisfied with probability 1, whereas the upper bound is satisfied with probability $1-\varepsilon$.}
\label{fig:constrExample}
\end{figure}

The example shows that jointly optimizing multiple parameters to determine offline a minimal size confidence region might lead to overly conservative results equal to the deterministic robust program.
Finally, consider the case where the probability density function of $e$ is single valued over some region, e.g. uniform disturbance model. In this case both optimizations in Example~\ref{ex:jointCCParamDet} might not have a unique optimizer, a standard problem in determining a confidence region. Hence the conservativeness of the resulting deterministic problem depends on the chosen optimizer and initial value.

We conclude this section emphasizing that direct tightening of single chance constraints gives the best worst-case value in terms of conservativeness of approximating a joint chance constraint by means of offline probability calculation.

\begin{rem}
  If the constraint tightening can be given as a function of the violation probability $\varepsilon_j$, then by using Boole's Inequality and dynamic risk allocation~\cite{Ono2008_IterativeRiskAllocationMPCwithJointCC} the conservatism can be further reduced at the cost of higher online computations.
\end{rem}

\section{Implementation and Numerical Example} \label{sec:NumExample}
In this section, we briefly review practical considerations for solving the single chance constrained programs to determine the proposed constraint tightening. Thereafter, the non-conservativeness of the approach with respect to the allowed probability of constraint violation and the increased feasible region is demonstrated in a numerical example.

\subsection{Solving the Single Chance Constrained Programs} \label{ssec:SolvingSingleCC}
There is a vast literature on how to (approximately) solve optimization programs involving single chance constraints,\cite{prekopa2010_StochasticProgramming}. In the following, we briefly state deterministic, as well as sampling based solutions to efficiently solve the offline problems~\eqref{eqn:offlineChanceConstrProgr},~\eqref{eqn:offlineChanceConstrProgrU} and~\eqref{eqn:offlineChanceConstrProgrTermConstr}.
\subsubsection{Deterministic}
Chance constraints are constraints on multivariate integrals. In particular, if the random variable $W_k$ has a known probability density function $f_W(w)$, we can write~\eqref{eqn:offlineChanceConstrProgr} as
\begin{equation*}
  \begin{aligned}
      &{[\eta_{l}]_j} = \max_{\eta} \eta\\
      &\text{s.t. } \int_{\W^l} \mathbf{1}_{\left\{ \eta \le [h]_j - \left[ H \right]_j  e_{l|k}  \right\}} \prod_{i=0}^{l-1}f_W(w_i) \dd w_0 \cdots \dd w_{l-1} \ge 1-[\varepsilon]_j
  \end{aligned}
\end{equation*}
with $e_{l|k} = \sum_{i=0}^{l-1}A_{cl}^iB_w w_i$ and $\mathbf{1}_{\{\cdot\}}$ being the indicator function.
The multivariate integral can numerically be approximated by quadrature rules suitable for high-dimensional integrals like Quasi-Monte Carlo or Sparse Grid methods~\cite{Tempo2012_RandAlgForAnalysisAndDesign}. 
While this formulation allows for a direct solution with standard nonlinear or stochastic optimization solvers, it can be further simplified to 1-dimensional integrals if $f_W(w) = \prod_{s=1}^{m_w}f_{W_s}(w_{s})$, i.e. the individual random variables in the random vector $W_k$ are independent.
Let $\tilde h^{i} = [H]_j A_{cl}^{i} B_w$ and 
\begin{equation*}
  \begin{aligned}
    f_{\tilde h^{i}_{s}W_s}(w_s) &= \frac{1}{|[\tilde h^{i}]_{s}|}f_{W_s}\left(\frac{w_s}{[\tilde h^{i}]_{s}}\right),\\
  f_{\tilde h^{i}} &= f_{\tilde h^{i}_{1}W_s} \ast f_{\tilde h^{i}_{2}W_s} \ast \ldots \ast f_{\tilde h^{i}_{m_w}W_s},
  \end{aligned}
\end{equation*}
where $f\ast g$ denotes the convolution of $f$ and $g$.
The probability density function $f_{H_je_l}$ of $[H]_j e_{l|k}$ is then given by
\begin{equation*}
  f_{H_je_l} = f_{\tilde h^{0}} \ast f_{\tilde h^{1}} \ast \ldots \ast f_{\tilde h^{l-1}}
\end{equation*}
and 
\begin{equation*}
  \begin{aligned}
      &-{[\eta_{l}]_j} = \min_{\eta} \eta\\
      &\text{s.t. } \int_{-\infty}^{\eta} f_{H_je_l}(x+[h]_j) \dd x \ge 1-[\varepsilon]_j.
  \end{aligned}
\end{equation*}
This formulation involves only 1-dimensional integrals which can be easily evaluated numerically. Due to the multiple convolutions, it might be beneficial to work with the Fourier Transform of $f_{W_s}$ instead.

For further discussions on approximations and tailored numerical optimization schemes, see e.g.~\cite[Chapter 8]{prekopa2010_StochasticProgramming}.

\subsubsection{Sampling}
Recently, sampling techniques to solve robust and chance constrained problems have gained increased interest~\cite{Tempo2012_RandAlgForAnalysisAndDesign,Calafiore2011_ProbMethodsFrContrSysDesign}. They are independent of the underlying distribution, easy to implement and specific guarantees about their solution can be given. In particular, they allow to directly use complicated simulations or measurements of the error, instead of determining a probability density function. 

The chance constrained problems~\eqref{eqn:offlineChanceConstrProgr},~\eqref{eqn:offlineChanceConstrProgrU},~\eqref{eqn:offlineChanceConstrProgrTermConstr} can be efficiently solved to the desired accuracy by drawing a sufficiently large number $N_s$ of samples $w^{(i)}$ from $W$ and require the constraint to hold for all, but a fixed number $r$ of samples.
In~\cite{Campi2011_SampleAndDiscardApprTpCCOpt} the authors give, under the condition $\varepsilon_u N_s > k$, the explicit conditions
\begin{equation}
  \begin{aligned}
    r &\le\epsilon_u N_s-\sqrt{2\epsilon_u N_s \ln\frac{1}{\beta}}, \\
    r &\ge\epsilon_l N_s-1+\sqrt{3\epsilon_l N_s \ln\frac{2}{\beta}}
  \end{aligned}
  \label{eqn:CampiBound}
\end{equation}
to select $r$ and $N_s$ such that with confidence $1-\beta$ the solution to the sampled program is equal to the chance constrained programs~\eqref{eqn:offlineChanceConstrProgr},~\eqref{eqn:offlineChanceConstrProgrU} and~\eqref{eqn:offlineChanceConstrProgrTermConstr} with $\epsilon \in [\epsilon_l, \epsilon_u]$.

In general, one has to solve a mixed integer problem or use heuristics to discard samples in a (sub)optimal way. Here, due to the simple structure, a sort algorithm is used to solve the sampled approximation of~\eqref{eqn:offlineChanceConstrProgr},~\eqref{eqn:offlineChanceConstrProgrU},~\eqref{eqn:offlineChanceConstrProgrTermConstr} exactly.
\begin{prop}~\label{prop:etaSampling}
  Let $N_s$ and $r$ be chosen according to~\eqref{eqn:CampiBound}.
  Let $q_{1-r/N_s}$ be the $(1-r/N_s)$-quantile of the set $\left\{ [H]_j e_{l|k}^{(i)}\right\}_{i=1,\ldots,N_s}$ with ${e_{l|k}^{(i)} = \sum_{j=1}^l A_{cl}^{j-1}B_w w_j^{(i)}}$ independently chosen samples from $W^l$.
  Then with confidence $1-\beta$
  \begin{equation*}
    {[\eta_l]_j} = [h]_j - q_{1-r/N_s}
  \end{equation*}
  solves~\eqref{eqn:offlineChanceConstrProgr} with $\epsilon \in [\epsilon_l, \epsilon_u]$.
\end{prop}

If the tightened constraints are derived via a sampling approach, the results on chance constraint satisfaction do not hold with certainty but only with confidence $(1-\beta)^p$. 
Let $C_\PP = \left\{ u\in \R^m ~|~ \PP\left\{ [H]_j x_{k+1} \le [h]_j ~|~ x_k \right\}\ge 1-[\varepsilon]_j, ~ j=1,\ldots p \right\}$
and with $\eta_1$ derived through Proposition~\ref{prop:etaSampling} define $C_S = \left\{ v\in \R^m ~|~  H z_{k+1} \le \eta_1 \right\}$, then $\PP\left\{ C_S \subseteq C_{\PP} \right\} \ge (1-\beta)^p$, i.e. with at least probability $(1-\beta)^p$ the feasible set derived through sampling is a subset of the feasible set for the chance constraints.
Yet, note that when sampling is used to determine the constraint tightening, the closed-loop state transition is not Markovian any more and the results do not hold if chance constraints on $x_{k+1}$ given $x_0,\ldots, x_k$ are considered. The result on recursive feasibility still holds, as long as the true disturbance set does not exceed the assumed disturbance set $\W$.

\subsection{Numerical Example}
In the following, the performance and enlarged region of attraction of the proposed Stochastic MPC scheme is demonstrated.
To this end, the linear system of the form~\eqref{eqn:xsystem} with
\begin{equation*}
  A = \begin{bmatrix} 1 & 0.0075 \\ -0.143 & 0.996 \end{bmatrix}, \quad 
  B = \begin{bmatrix} 4.798 \\ 0.115 \end{bmatrix}, \quad 
  B_w = I_2,
\end{equation*}
is used as prediction and system model.%
\footnote{%
  In the context of linear Stochastic MPC, the system has previously been considered in~\cite{Cannon2011_StochasticTubesinMPC} as a linearized model of a DC-DC converter which, in the context of Nonlinear MPC, originally appeared in~\cite{Lazar2004_NonlinearPredContrOfDCDCconv,Lazar2008_ISSsuboptimalNMPCwithApplToDCDCconv}.}

The SMPC cost weights are $Q = \begin{bmatrix} 1 & 0 \\ 0 & 10 \end{bmatrix}$, $R=1$ and the prediction horizon is $T=8$. For disturbance attenuation in the predictions
, the unconstrained LQR is chosen.
The disturbance distribution is assumed to be a truncated Gaussian with the covariance matrix $\Sigma = 0.04^2 I_2$ truncated at $\|w\|^2 \le 0.02$.

For the robust set calculations 
we chose a polytopic outer approximation with 8 hyperplanes.
For the stochastic constraint tightening we used the described sampling approach with $\epsilon_l = 0.95\epsilon$, $\epsilon_u = 1.05\epsilon$ and confidence $\beta = 10^{-4}$. 
The constraint tightening was performed without explicitly bounding $\varepsilon_f$, the probability of the candidate solution not being feasible, as described in Section~\ref{ssec:feasCandidateSol}.

The time for computing the SMPC input using \texttt{quadprog} with the standard interior point algorithm in Matlab R2014b was approximately 4ms for each scheme on an Intel Core i7 with 3.4GHz.

\subsubsection*{Constraint Violation}
First, consider the single chance constraint
\begin{equation}
  \PP_k\left\{ [x]_1 \le 2 \right\} \ge 0.8
  \label{eqn:ex1Constraint}
\end{equation}
for the system above with initial state $x_0 = [2.5 ~ 2.8]^\top$.

In~\cite{Cannon2011_StochasticTubesinMPC} it has been shown that Stochastic MPC achieves lower closed-loop cost compared to Robust MPC. The approach presented in~\cite{Cannon2011_StochasticTubesinMPC}, using a confidence region, yields $14.4\%$ constraint violation in the first $6$ steps.

In contrast, the approach taken here, i.e. a direct constraint tightening, achieves a closed-loop operation tight at the constraint. A Monte Carlo simulation with $10^4$ realizations showed an average constraint violation in the first 6 steps of $20\%$ and an even lower closed-loop cost. Simulation results of the closed-loop system for $100$ random disturbances are shown in Figure~\ref{fig:SimResBoth}. The left plot shows the complete trajectories for a simulation time of $15$ steps. The right plot shows the constraint violation in more detail,~\eqref{eqn:ex1Constraint} is satisfied non-conservatively hence leaving more control authority for optimizing the performance.
\begin{figure}[htpb]
  \subfloat{{ 
    \includegraphics{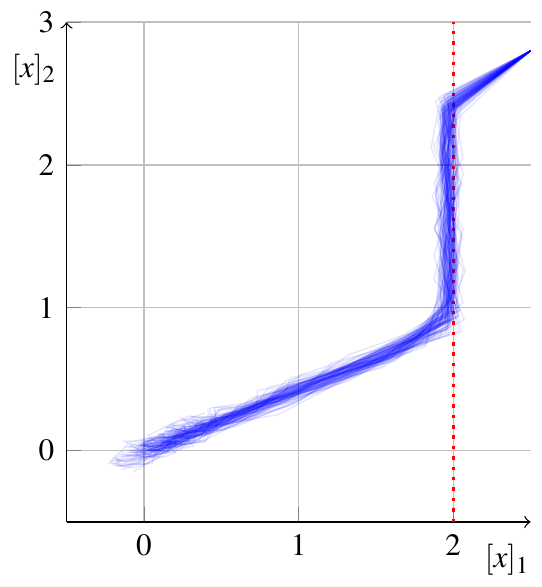}
    }}
  \hspace{-0.3cm}
  \subfloat{{ 
    \includegraphics{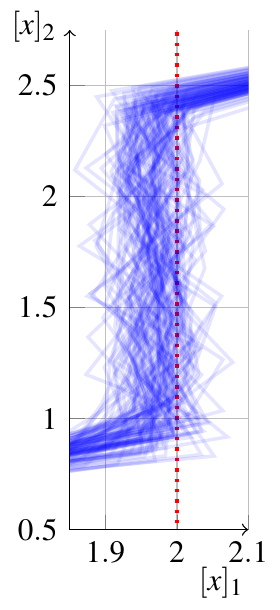}
  }}
  \caption{\small Left: Plot of closed-loop response with 100 different disturbance realizations and initial state $x_0 = [2.5 ~ 2.8]^\top$. \\
    Right: Detail showing the trajectories near the constraint ${\PP\{ [x]_1 \le 2\} \ge 0.8}$.
A Monte Carlo simulation with $10^4$ realizations showed an average constraint violation in the first 6 steps of $20\%$. }
  \label{fig:SimResBoth}
\end{figure}

For comparison, we remark that Robust MPC achieves $0\%$ constraint violation and that the LQ optimal solution violates the constraint $100\%$ in the first $3$ steps.

\subsubsection*{Feasible Region}
The main advantage of the proposed SMPC scheme is the increased feasible region.

To illustrate this feature, we assume the same setup as before, but with additional chance constraints on the state 
and hard input constraints
\begin{equation*}
  \begin{aligned}
    \PP_k\{~ [x]_1 &\le 2 ~\} \ge 0.8, & \PP_k\{~ -[x]_1 &\le 2 ~\} \ge 0.8, \\
    \PP_k\{~ [x]_2 &\le 3 ~\} \ge 0.8, & \PP_k\{~ -[x]_2 &\le 3 ~\} \ge 0.8, \\
              |u| &\le 0.2.
  \end{aligned}
\end{equation*}
According to the described setup, we allowed $5\%$ constraint violation in the predictions for the input and a probability of $0.05$ of not reaching the terminal region. In closed-loop operation the input was treated as hard constraint.

Figure~\ref{fig:CompFeasReg} shows the different feasible regions of Robust MPC, Stochastic MPC with constraint tightening using recursively feasible probabilistic tubes and the proposed method using probabilistic tubes and a first step constraint.
The feasible region of the proposed Stochastic MPC has $1.7$ times the size of the feasible region of \emph{standard} SMPC and $3.4$ times the size of the feasible region of Robust MPC. The Robust MPC scheme has been taken from~\cite{Mayne2005_RobustMPCofConstrLinSysWithBoundedDist} and is only included here for a more complete comparison, it is of course significantly smaller than having stochastic constraints.

In Figure~\ref{fig:SizeFeasReg} the decrease in the size of the feasible region is plotted, when a constraint tightening as described in Section~\ref{ssec:feasCandidateSol} is employed. Note that even for moderate values of $\varepsilon_f$ a significant increase in the feasible region can be gained.

\begin{figure}[htpb]
  \begin{center}
    \includegraphics{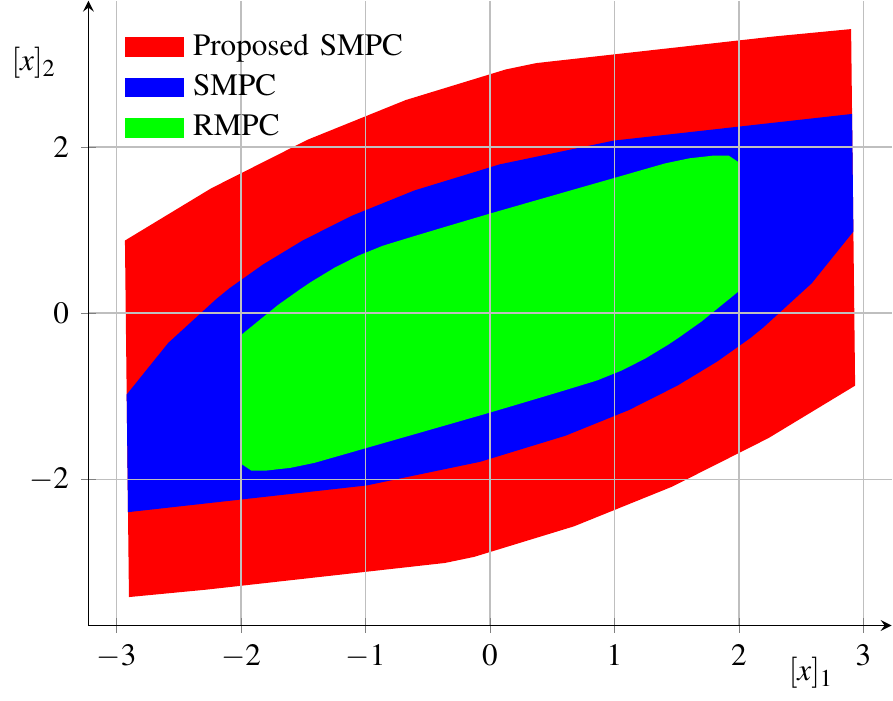}
  \end{center}
  \caption{\small Comparison of feasible region for Robust MPC, Stochastic MPC with recursively feasible probabilistic tubes and proposed Stochastic MPC with guaranteed recursive feasibility.}
  \label{fig:CompFeasReg}
\end{figure}

\begin{figure}[htpb]
  \begin{center}
    \includegraphics{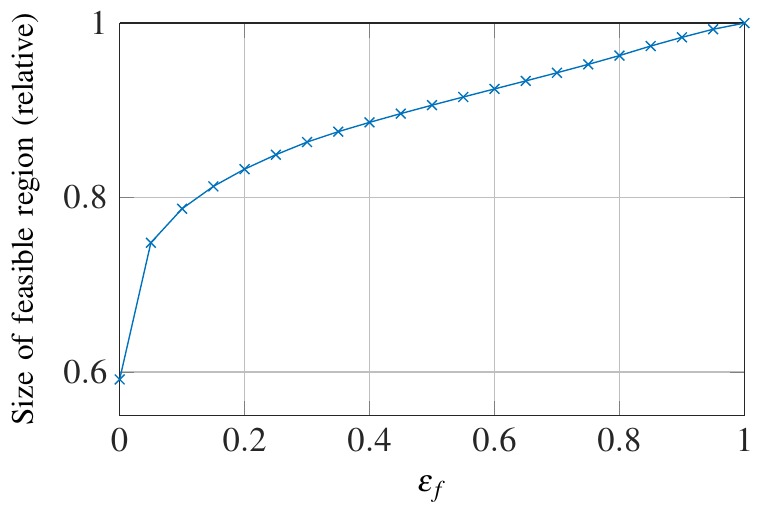}
  \end{center}
  \caption{\small Relative size of the feasible region plotted over $\varepsilon_f$, the maximal probability of the candidate solution not remaining feasible.}
  \label{fig:SizeFeasReg}
\end{figure}

\section{Conclusions and Future Work} \label{sec:Concl}
The proposed stabilizing Stochastic MPC algorithm provides a significantly increased feasible region through separating the requirements of recursive feasibility and stability.
The algorithm unifies the results obtained in~\cite{Kouvaritakis2010_ExplicitUseOfProbConstr} and~\cite{Korda2011_StronglyFeasibleSMPC} allowing to balance convergence speed and performance guarantees against the size of the feasible region. 
Absolute bounds of the disturbance are used to provide a first step constraint to guarantee robust recursive feasibility.
The stochastic information about the disturbance is used to prove an asymptotic bound on the closed-loop performance, which naturally resembles the bound obtained by the unconstrained LQ-optimal controller.
Furthermore, under mild assumptions, asymptotic stability with probability one of the set $\X_\infty$ has been proven, which is novel in the Stochastic MPC literature.

The online computational effort is equal to that of nominal MPC.
An efficient, broadly applicable solution strategy based on randomized algorithms is presented to solve, to the desired accuracy, the offline chance constrained problems for determining the constraint tightening.

Future work will be focused on improving the performance through an online evaluation of the expected cost, taking into account possible infeasibility of the optimized input trajectory. 
Similarly, the idea to incorporate a first step constraint to guarantee recursive feasibility could be further exploited. In the future this could be applied in a broader context, e.g. 
it could be nicely combined with ideas of (incomplete) decision trees which show very good results in practice~\cite{Lucia2013_MultiStageNMPCforSemiBatchReactor}, but have no recursive feasibility or stability guarantees.
For a broader applicability, it is necessary to relax the assumption of identically and independently distributed disturbance as well as allow for parametric uncertainty.
Finally, for unbounded additive disturbances, possible approaches could be a similar constraint tightening combined with a suitable penalty reformulation of the state constraints.

\bibliographystyle{IEEEtranNoUrl}
\bibliography{IEEEabrv,GlobalBib}
\end{document}